\documentclass{article}

\usepackage{amssymb}
\usepackage[italicdiff]{physics}
\usepackage{mathtools}
\mathtoolsset{showonlyrefs}
\usepackage{listings}
\usepackage{refcount}
\usepackage[shortlabels]{enumitem}
\usepackage{amsthm}
\usepackage[noend]{algorithmic}
\usepackage[hyphens]{url}

\let\mc\mathcal
\let\eps\varepsilon
\let\op\operatorname

\newcommand{\R}{\mathbb R}

\newcommand{\LP}{{\op{LP}}}

\let\norm\undefined
\DeclarePairedDelimiter{\norm}{\lVert}{\rVert}

\DeclareMathOperator*{\argmin}{argmin}

\newtheorem{theorem}{Theorem}
\newtheorem*{theorem*}{Theorem}
\newtheorem*{lemma*}{Lemma}

\theoremstyle{definition}

\newtheorem{exampl}{Example}

\usepackage{microtype}
\usepackage{graphicx}
\usepackage{subfigure}
\usepackage{booktabs} %

\usepackage[pdfusetitle]{hyperref}

\usepackage[accepted]{icml2020}

\icmltitlerunning{Sparsified Linear Programming for Zero-Sum Equilibrium Finding}

\begin{document}
	
	\twocolumn[
	\icmltitle{Sparsified Linear Programming for Zero-Sum Equilibrium Finding}

	\icmlsetsymbol{equal}{*}
	
	\begin{icmlauthorlist}
		\icmlauthor{Brian Hu Zhang}{cmu}
		\icmlauthor{Tuomas Sandholm}{cmu,sm,sr,om}
	\end{icmlauthorlist}
	
	\icmlaffiliation{cmu}{Department of Computer Science, Carnegie Mellon University, Pittsburgh, PA, USA}
	\icmlaffiliation{sm}{Strategic Machine, Inc.}
	\icmlaffiliation{sr}{Strategy Robot, Inc.}
	\icmlaffiliation{om}{Optimized Markets, Inc.}
	
	\icmlcorrespondingauthor{Brian Hu Zhang}{bhzhang@cs.cmu.edu}
	\icmlcorrespondingauthor{Tuomas Sandholm}{sandholm@cs.cmu.edu}
	
	\icmlkeywords{Game Theory and Mechanism Design}
	
	\vskip 0.3in
	]

	\printAffiliationsAndNotice{}  %
	
	\begin{abstract}
		Computational equilibrium finding in large zero-sum extensive-form imperfect-information games has led to significant recent AI breakthroughs. The fastest algorithms for the problem are new forms of counterfactual regret minimization~\cite{bs19}. In this paper we present a totally different approach to the problem, which is competitive and often orders of magnitude better than the prior state of the art. The equilibrium-finding problem can be formulated as a linear program (LP)~\cite{kmv94}, but solving it as an LP has not been scalable due to the memory requirements of LP solvers, which can often be quadratically worse than CFR-based algorithms. We give an efficient practical algorithm that factors a large payoff matrix into a product of two matrices that are typically dramatically sparser. This allows us to express the equilibrium-finding problem as a linear program with size only a logarithmic factor worse than CFR, and thus allows linear program solvers to run on such games. With experiments on poker endgames, we demonstrate in practice, for the first time, that modern linear program solvers are competitive against even game-specific modern variants of CFR in solving large extensive-form games, and can be used to compute exact solutions unlike iterative algorithms like CFR.
	\end{abstract}
	
	\section{Introduction}
	Imperfect-information games model strategic interactions between agents that do not have perfect knowledge of their current situation, such as  auctions, negotiations, and recreational games such as poker and battleship. 
	Linear programming (LP) can be used to solve---that is, to find a Nash equilibrium in---imperfect-information two-player zero-sum perfect-recall games~\cite{kmv94}. However, due mostly to memory usage (see e.g.,~\citealp{zinkevich07} or~\citealp{bs19}), it has generally been thought of as impractical for solving large games. Thus, a series of other techniques has been developed for solving such games. Most prominent among these is the {\it counterfactual regret minimization} (CFR) family of algorithms~\cite{zinkevich07}. These algorithms work by iteratively improving both player's strategies until their time averages converge to an equilibrium. They have a worst-case bound of $O(1/\eps^2)$ iterations needed to reach accuracy $\eps$, and more recent improvements, most notably {\it CFR+}~\cite{tammelin14} and {\it discounted CFR} (DCFR)~\cite{bs19} mean that algorithms from the CFR family remain the state of the art in practice for solving large games, and have been used as an important part of the computational pipelines to achieve superhuman performance in benchmark cases such as heads-up limit~\cite{bbjt17} and no-limit~\cite{bs18} Texas hold'em.
	
	Several families of algorithms have theoretically faster convergence rates than those of the CFR family. First-order methods~\cite{hgps10,kwks15} have a theoretically better convergence guarantee of $O(1/\eps)$ (or even $\log(1/\eps)$ with a condition number~\cite{Gilpin12:First}), but in practice perform worse than the newest algorithms in the CFR family~\cite{Kroer18:Solving,bs19}. Standard algorithms for LP are known that converge at rate $O(\log(1/\eps))$, but for the most part, these algorithms require storage of the whole payoff matrix explicitly, which the CFR family does not, and often require superlinear time per iteration (with respect to the number of nonzero entries of the LP constraint matrix), which is prohibitive when the game is extremely large. 
	
	In this paper we investigate how to reduce the space requirements of LP solvers by factoring the possibly dense payoff matrix of an extensive-form game. A long body of work investigates the problem of decomposing, or factoring, a given matrix $A$ as the product of other matrices, with some objective in mind. Studied objectives include the speedup of certain operations, as in the LU or Cholesky factorizations, and approximation of the matrix $A$ in a certain norm, as in the singular value decomposition (SVD). Our objective in this work is {\it sparsity}: we investigate the problem of factoring a matrix $A$ into the product of two matrices $U$ and $V$ that are much sparser than $A$. This differs from the usual low-rank approximation in that the optimization objective is different ($0$-norm, that is, number of non-zero entries, instead of the $2$-norm, that is, the square root of sum of squares), and that the matrices $U$ and $V$ might not be low rank (and in fact in poker they will high rank that is linear in the number of sequences in the game). 
	
	We are not aware of any prior application-independent work that addresses this problem. The SVD approximates $A$ in the wrong norm for this purpose: $2$-norm approximations will in the general case have fully dense residual, which is not desirable. The body of work on {\it sparse PCA} (e.g.,~\citealp{zx18}) focuses on {\it low-rank} sparse factorizations. That still mostly focuses on $2$-norm approximations, and the runtime of the algorithm usually scales with the rank of the factorization as well as the size of $A$. In our cases, an optimal factorization may have {\it high} or even full rank (and yet be sparse), and our matrices are large enough that quadratic (or worse) dependence on $\norm{A}_0$, which is often seen in these algorithms, is unacceptable. Our goal is to find such a factorization efficiently.~\citet{np13} address the same problem, but restrict their attention to matrices that are known {\it a priori} to be the product of sparse matrices with entries drawn independently from a nice distribution. This is not the case in our setting.~\citet{rov14} attack a related but still substantially different problem, of finding a sparse factorization when we know {\it a priori} a good bound on the sparsity of each row or column of the factors. This, too, is not true in our setting: no such bound may even exist, much less be known.
	
	Our main technical contribution is a novel practical matrix factorization algorithm that greatly reduces the size of the game payoff matrix in many cases. This matrix factorization allows LP algorithms to run in far less memory and time than previously known, bringing the memory requirement close to that of CFR. We demonstrate in practice that this method can reduce the size of a payoff matrix by multiple orders of magnitude, yielding improvements in both the time and space efficiency of solving the resulting LP. This makes our approach---automated matrix sparsification followed by LP---superior to domain-independent versions of the fastest CFR variant. If high accuracy is desired, our domain-independent approach in many cases outperforms even a highly customized poker-specific implementation of the fastest CFR variant~\citep{bs19}. 
	
	We show experiments with the primal simplex, dual simplex, and the barrier method as the LP solver. The barrier method runs in polynomial time but each iteration is heavy in terms of memory and time. For that reason, we present techniques that significantly speed up a recent $O(\log^2(1/\eps))$ LP algorithm~\citep{yen15} that has iteration time and memory linear in the number of nonzero entries in the constraint 
	matrix, and show experiments with that as well. Our experiments show interesting performance differences among the LP solvers as well.

	\section{Preliminaries}
	\noindent {\bf Extensive-form games. } We study the standard representation of games which can include sequential and simultaneous moves, as well as imperfect information, called an {\it extensive-form game}. It consists of the following.
	(1) A set of players $\mc P$, usually identified with positive integers. Random chance, or ``nature'' is also considered a player, and will be referred to as player 0.
	(2) A finite tree $H$ of {\it histories} or {\it nodes}, rooted at some {\it initial state} $\emptyset \in H$. Each node is labeled with the player (possibly nature) who acts at that node. The set of leaves, or {\it terminal states}, in $H$ will be denoted $Z$. The edges connecting any node $h \in H$ to its children are labeled with {\it actions}. 
	(3) For each player $i \in \mc P$, a {\it utility function} $u_i : Z \to \R$. 
	(4) For each player $i \in \mc P$, a partition of the nodes at which player $i$ acts into a collection $\mc I_i$ of {\it information sets}. In each information set $I \in \mc I_i$, every pair of nodes $h, h' \in I$ must have the same set of actions.
	(5) For each node $h$ at which nature acts, a distribution $\sigma_0(h)$ over the actions available that node.
	
	For any history $h \in H$ and any player $i \in \mc P$, the {\it sequence} $h[i]$ of player $i$ at node $h$ is the sequence of information sets reached and actions played by player $i$ on the path from the root node to $h$.  The set of sequences for player $i$ is denoted $S_i$. A player $i$ has {\it perfect recall} if $h[i] = h'[i]$ whenever $h$ and $h'$ are in the same information set $I \in \mc I_i$. In this work, we will focus our attention on two-player zero-sum games of perfect recall; i.e., games in which $\mc P = \{1, 2\}$, $u_1 = -u_2$, and both players have perfect recall. For simplicity of notation, the opponent of player $i$ will be denoted $-i$. 
	
	A {\it behavior strategy} (hereafter simply {\it strategy}) $\sigma_i$ for player $i$ is, for each information set $I \in J_i$ at which player $i$ acts, a distribution $\sigma_i(I)$ over the actions available at that infoset. When an agent reaches information set $I$, it chooses an action according to $\sigma_i(I)$. A pair $(\sigma_1, \sigma_2)$ of behavior strategies, one for each player, is a {\it strategy profile}. The {\it reach probability} $\pi^\sigma_i(h)$ is the probability that node $h$ will be reached, assuming that player $i$ plays according to strategy $\sigma_i$, and all other players (including nature) always choose actions leading to $h$ when possible. This definition extends to sets of nodes or to sequences by summing the reach probabilities.
	
	The {\it best response value} $\op{BRV}(\sigma_{-i})$ for player $i$ against an opponent strategy $\sigma_{-i}$ is the largest achievable value; i.e., in a two-player game, $\op{BRV}(\sigma_{-i}) = \max_{\sigma_i} u_i(\sigma_i, \sigma_{-i})$. A strategy $\sigma_i$ is an $\eps$-{\it  best response} to opponent strategy $\sigma_{-i}$ if $u_i(\sigma_i, \sigma_{-i}) \ge \op{BRV}(\sigma_{-i}) - \eps$. A strategy profile $\sigma$ is an $\eps$-{\it Nash equilibrium} if its {\it Nash gap} $\op{BRV}(\sigma_2) + \op{BRV}(\sigma_1)$ is at most $\eps$. {\it Best responses} and {\it Nash equilibria} are respectively $0$-best responses and $0$-Nash equilibria. The {\it exploitability} $\exp(\sigma_i)$ of a strategy is how far away $\sigma_i$ is away from a Nash equilibrium: $\exp(\sigma_i) = \op{BRV}(\sigma_i) - \op{BRV}(\sigma_i^*)$ where $\sigma_i^*$ is a Nash equilibrium strategy for the player. In a zero-sum game, the Nash value $\op{BRV}(\sigma_i^*)$ is the same for every Nash equilibrium strategy, so the exploitability is well-defined. 
	
	{\bf Equilibrium finding via linear programming.}
	Nash equilibrium finding in an extensive-form game can be cast as an LP in the following fashion~\citep{vonstengel96}. Consider mapping behavior strategies $\sigma_i$ to vectors $x \in \R^{S_i}$ by setting $x(s) = \pi^\sigma_i(s)$ for every sequence $s$. 
	We will refer to vector $x$ as a strategy. Under this framework, equilibrium finding can be cast as a bilinear saddle point problem
	\begin{align}
	\max_{x \ge 0} \min_{y \ge 0}\ & x^T A y \quad \text{s.t. } Bx = b,\quad Cy = c,\quad x, y \ge 0
	\end{align}
	where the matrices $B$ and $C$ satisfy $\norm{B}_0 = O(\abs{S_1})$, $\norm{C}_0 = O(\abs{S_2})$, and encode the constraints on the behavior strategies $x$ and $y$. $A$ is the payoff matrix whose $(i, j)$ entry is the expected payoff for Player 1 when Player 1 plays to reach sequence $i$ and Player 2 plays to reach sequence $j$: 
	$
	A = \sum_{z \in Z} \pi_0(z) u_1(z[1], z[2]) e_{z[1]} e_{z[2]}^T
	$
	where $e_i$ is the $i$th unit vector. The number of entries $\norm{A}_0 \le \abs{Z}$. Now taking the dual of the inner minimization yields the LP
	\begin{align}\label{eq:lp}
	\max_{x \ge 0,z} c^T z \qq{s.t.} \quad Bx = b,\quad C^T z \le A^T x.
	\end{align}
	Expressed in any LP standard form, the constraint matrix has $
	O(\abs{S_1} + \abs{S_2} + \abs{Z})
	$ nonzero entries in its constraint matrix. The LP can be solved with any standard solver.
	
	{\bf Sparse linear programming.} \citet{yen15} give a generic algorithm for solving LPs in the standard form\footnote{
		We use the subscript LP everywhere due to the clash of variable naming conventions between LP (where $A_\LP$ is the constraint matrix) and equilibrium finding (where $A$ is the payoff matrix).}
	\begin{align}
	\min_{ x_\LP \ge 0} \quad c^T_\LP x_\LP \qq{s.t.} A_\LP x_\LP \le b_\LP \label{eq:lp-standard}
	\end{align} 
	We first give a brief overview of the algorithm.
	We are interested in LPs of the standard form~\eqref{eq:lp-standard}
	and their duals
	\begin{equation}\label{eq:lp-dual}
	\min_{y_\LP\ge 0} \quad b_\LP^T y_\LP \qq{s.t.} {-}A_\LP^Ty_\LP \le c_\LP
	\end{equation}
	where $A \in \R^{m \times n}$. Consider the convex subproblem
	\begin{align}
	\min_{y_\LP\ge 0}\ & b_\LP^T y_\LP + \frac{1}{2\eta} \norm{y_\LP - \hat y}_2^2 \quad\text{s.t. } {-}A_\LP^Ty_\LP \le c_\LP
	\end{align}
	for some given initial solution $\hat y \in \R^m$ and real number $\eta > 0$. 
	The dual of this subproblem is
	\begin{align}
	\min_{x,z}\quad& c_\LP^T x_\LP + \frac{\eta}{2} \norm*{A_\LP x_\LP - b_\LP + z_\LP +\frac{1}{\eta} \hat y}_2^2 \\ \qq{s.t.} & x_\LP\ge 0, z_\LP\ge 0\label{eq:lp-subproblem-dual}
	\end{align}

	The approach is shown in Algorithm~\ref{alg:alm}. 
	In Line 2, the solution to Problem~\eqref{eq:lp-subproblem-dual} is computed via either a randomized coordinate descent (RC) or a projected Newton-CG (PG) algorithm; the details are not important here. The breakthrough of~\citet{yen15} is an implementation of these inner loops in $O(\norm{A_\LP}_0)$ time, rather than $O(mn)$ or worse. At each iteration, $x^*$ is infeasible in the original problem since the quadratic regularization term in~\eqref{eq:lp-subproblem-dual} does not punish slightly infeasible solutions much at all. $y^*$ is  infeasible since $(x^*, z^*)$ is a suboptimal solution to~\eqref{eq:lp-subproblem-dual}. Thus, Algorithm~\ref{alg:alm} works with infeasible solutions to the LP, which must be projected back into the feasible space.
	\begin{algorithm}[!ht]
		\caption{Augmented Lagrangian algorithm for solving linear programs \protect\cite{yen15}}\label{alg:alm}
		\textbf{Input:} initial dual solution guess $\hat y \in \R^m$, parameter $\eta > 0$
		
		\textbf{Output:} primal-dual solution pair $(x^*, \hat y)$
		
		\begin{algorithmic}[1]
			\LOOP
			\STATE let $(x^*, z^*)$ be an approximate solution to 
			\\\quad Problem~\eqref{eq:lp-subproblem-dual} given the current $\hat y$ and $\eta$.
			\STATE set $\hat y \gets \hat y + \eta (A_\LP x^* - b_\LP + z^*)$
			\STATE if necessary (as detailed by~\citet{yen15}),
			\\\quad increase $\eta$ by a constant factor
			\ENDLOOP
		\end{algorithmic}
	\end{algorithm}
	\begin{theorem}[Theorem 3 in~\citealp{yen15}]
		\label{thm:slp-main}
		After $O(\log(1/\eps))$ outer iterations of Algorithm~\ref{alg:alm}, each of which is run for $O(\log(1/\eps))$ inner iterations, we have $d(\hat y, S) \le \eps$ where $S\subseteq \R^m$ is the set of dual-optimal solutions and $d$ is Euclidean distance. 
	\end{theorem}
	
	The $O$ in the above theorem hides problem-dependent constants such as condition numbers. This theoretical guarantee applies to the dual solution, and not the primal. Thus, to find a primal-dual solution pair, in theory we must run Algorithm~\ref{alg:alm} twice: on the primal (to find a dual solution) and then the dual (to find a primal solution). In practice, the primal solution from the first run already has extremely low exploitability, so the second run would be unnecessary. 
	
	The rest of the paper covers our new contributions.
	
	\section{Adapting the $O(\log^2(1/\eps))$ Sparse LP Solver}
	
	In order to make the above LP algorithm fast for game solving, we had to make a modification and also deal with the caveat of eternally infeasible solutions $x_\LP$ and $y_\LP$.
	\subsection{Limiting the Number of Inner Iterations}
	
	\citet{yen15} give an implementation of their algorithm, which they call {\it LPsparse}. In it, the inner loop runs until either (1) it converges to a sufficiently small error tolerance, or (2) some prescribed iteration limit is hit. The iteration limit is set to increase exponentially every time it is hit. In practice, we found this to be far too aggressive, leading to inner loops that take prohibitively long (an hour or more on two-player no-limit Texas hold'em endgames). Thus, we instead we only allow the number of inner iterations to grow linearly with respect to the number of outer iterations. Since both the outer and inner loop lengths are bounded by $O(\log(1/\eps))$ in  Theorem~\ref{thm:slp-main}, this does not change the theoretical performance guarantee of the algorithm, and it leads to a significant speedup in practice. 
	
	\subsection{Normalizing Infeasible Solutions}\label{s:normalization}
	Algorithm~\ref{alg:alm} will output an infeasible solution pair. To retrieve a valid behavior strategy (feasible solution), we first project into the positive orthant (i.e., zero out any negative entries), and then normalize each information set in topological order, starting with the root. This results in a strategy pair whose Nash gap we can evaluate. This normalization step roughly maintains the guarantee of Theorem~\ref{thm:slp-main}:
	\begin{theorem}\label{thm:normalization}
		Suppose $x_\LP = (x, z)$ is an infeasible solution to~\eqref{eq:lp} such that $d((x, z), S) \le \eps$, where $S$ is the set of optimal solutions to $\eqref{eq:lp}$. Then the above normalization yields a (feasible) strategy with exploitability at most $\eps n^4 \norm{A}_\infty$, where $n$ is the total number of sequences between the two players.
	\end{theorem}
	A proof is in the appendix. The above bound is loose, but it is unnecessary to improve it for the theoretical punchline: combining Theorems~\ref{thm:slp-main} and~\ref{thm:normalization}, we see that the LP algorithm converges to a strategy with exploitability $\eps$ in $O(\log^2(1/\eps))$ inner iterations (where the $O$ possibly hides problem-dependent constants), assuming $A$ is normalized (i.e., $\norm{A}_\infty$ is fixed to, say, 1). 
	
	\section{Sparse Factorization} \label{s:factor}
	In many games, the payoff matrix $A$ is somewhat dense. This occurs when the number of terminal game tree nodes, $\abs{Z}$, is large compared to the total number of sequences $\abs{S_1} + \abs{S_2}$, that is, when a significant fraction of the sequence pairs represent valid terminal nodes. In most normal-form (a.k.a. matrix-form) games, $A$ is fully dense, whereas in extensive-form games of perfect information, $A$ is extremely sparse (because the number of terminal nodes equals the total number of terminal sequences between the players). In most real games, the value of each entry $A_{ij}$ can be computed in constant time from the indices $i$ and $j$ alone based on the rules of the game, with a minimal amount of auxiliary memory, so $A$ can be stored implicitly. In these cases, the sparse LP solver is at a disadvantage compared to the CFR family of algorithms. CFR can run with only implicit access to $A$. Its memory usage is thus $O(|S_1| + |S_2|)$. LP solvers, on the other hand, require a full description of $A$, which here will have size $O(|Z|)$. Our idea here is to make LP solvers practical by carefully compressing $A$ in a way that standard LP solvers can still handle. 
	
	This leads to our main idea. If we can write $A$ approximately as the product of two matrices; that is, $A = \hat A + UV^T$, such that $\norm{U}_0 + \norm{V}_0 + \norm{\hat A}_0 \ll \norm{A}_0$, then we 
	can reformulate 
	the LP~\eqref{eq:lp} as 
	\begin{align}
	\max_{x \ge 0, z, w} &\quad c^T z 
	\\ \qq{s.t.} &\quad  Bx = b,
	\quad C^T z \le V w + \hat A^T x,
	\quad U^T x = w
	\end{align}
	which, in standard form, has $O(\norm{B}_0 + \norm{C}_0 + \norm{U}_0 + \norm{V}_0 + \norm{\hat A}_0)$ nonzero constraint matrix entries. In this formulation, we demand that not only $U$ and $V$ but also the residual $\hat A$ be sparse. Depending on the density of $A$ and the quality of the factorization $\hat A + UV^T$, a good factorization could yield a quadratic improvement in both the time and space used by the LP solver.
	
	When $A$ is low-rank, SVD would provide such a factorization. However, in many cases, $A$ is not sparse and not well approximated by a low-rank factorization. Further, even when $A$ is low-rank, it is possible that, for example, $A - uv^T$ is a dense matrix (where $uv^T$ is the best rank-1 approximation to $A$), which means that the algorithm would take $\Omega(mn)$ time and memory per iteration starting from the second outer iteration. 
	We now give examples of matrices $A$ for which finding a sparse factorization in our style is superior to finding a standard low-rank factorization (SVD), both in speed and resulting sparsity. An additional example can be found in the appendix.
	\begin{exampl}\label{ex:rank1}
		Let $A_1 = uv^T$ be a rank-one matrix, and let $A$ be $A_1$, except its lower-triangular half has been zeroed out. In general, $A$ will now be full-rank, and the SVD of $A$ will not be sparse. However, we can express $A = UV^T$ with $\norm{U}_0 = \norm{V}_0 = O(n \log n)$ as follows. Set $u_0 = u$ except with its right half zeroed out, and set $v_0 = v$ except with its left half zeroed out. Then $u_0v_0^T$ matches the upper-right quadrant of the matrix $A$, as shown in Figure~\ref{fig:ex-rank1}. 
		Moreover, $A - u_0v_0^T$ is block diagonal, where the two blocks have size $(n/2) \times (n/2)$ and have the same structure as $A$ itself. Thus, we may recursively factor the two blocks. The vectors $u_0$ and $v_0$ both have $n/2$ nonzero entries, so the total number of nonzero entries in the factorization is expressed by the recurrence $S(n) = n + 2S(n/2)$, which solves to $S(n) = O(n \log n)$. The matrices $U$ and $V$ will both have $\Theta(n)$ columns. This example shows up in practice; the payoff matrix of poker endgames is block diagonal, where the blocks have essentially this form.
		\begin{figure}[!ht]
			\centering
			\includegraphics[width=0.25\columnwidth]{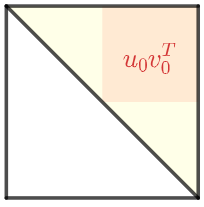}
			\caption{Illustration of factorization in Example~\ref{ex:rank1}. The box represents the matrix $A$. The upper right shaded regions represents its nonzero entries. The first iteration of the factorization zeros out the orange shaded box.}
			\label{fig:ex-rank1}
		\end{figure}
	\end{exampl}
	\begin{exampl}
		Let $A_0 = \hat A + UV^T$ be a sparsely-factorable matrix, and $A = A_0 + \hat A_0$ where the residual $\hat A_0$ may be high-rank, but is sparse. For example, perhaps $A$ is $A_0$ with some entries around its diagonal zeroed out. Then $A$ itself is also sparsely factorable as $A = (\hat A + \hat A_0) + UV^T$. This example may seem trivial, but the SVD does not share a similar property. For example, if $\hat A_0$ is the matrix from Example~\ref{ex:rank1}, and $A_0$ is a general sparsely-factorable matrix (even the zero matrix), then the SVD of $A = A_0 + \hat A_0$ will be dense, but $A$ will still have a sparse factorization.
	\end{exampl}
	
	\section{Factorization Algorithm}
	In this section, we develop a general algorithm for factoring an arbitrary sparse matrix $A$ into the product of two possibly sparser---and never denser---matrices. For this section, we let $m = |S_1|$ and $n = |S_2|$ so that $A \in \R^{m \times n}$. We follow the general strategy used by the power iteration SVD algorithm (e.g., \citealp{golub96matrix}). Algorithm~\ref{alg:factor} 
	reduces the factorization problem to solving, for a given matrix $A$, the subproblem 
	\begin{align}
	\argmin_{u,v} \norm{A - uv^T}. \label{eq:factor-sub}
	\end{align}
	\begin{algorithm}[!htb]
		\caption{Matrix factorization}\label{alg:factor}
		\textbf{Input:} matrix $A \in \R^{m \times n}$, norm $\norm{\cdot}$ on matrices
		
		\textbf{Output:} matrices $U \in \R^{m \times r}$ and $V \in \R^{n \times r}$
		
		\begin{algorithmic}[1]
			\STATE set $U$ and $V$ to be empty matrices
			\LOOP
			\STATE $u, v \gets \argmin_{u, v} \norm{A- uv^T}$
			\IF{$\norm{u}_0 > 1$ and $\norm{v}_0 > 1$} 
			\STATE $U \gets [U, u]$
			\STATE $V \gets [V, v]$
			\STATE $A \gets A - uv^T$
			\ENDIF
			\ENDLOOP
		\end{algorithmic}
	\end{algorithm}
	
	When $\norm{\cdot}$ is the 2-norm, this problem can be solved using the standard power iteration algorithm. However, when $\norm{\cdot}$ is the 0-norm, the problem is not so easy, and even using the 1-norm as a convex substitute for the 0-norm does not help:
	\begin{theorem}[\citealp{gillis2018complexity}]\label{thm:0-norm-hard}
		When $\norm{\cdot}$ is the 1-norm or 0-norm, the optimization problem~\eqref{eq:factor-sub} is NP-hard.
	\end{theorem}
	
	We thus resort to an algorithm that may not yield the optimal solution but works extremely well in practice. Algorithm~\ref{alg:factor-sub} reduces~\eqref{eq:factor-sub} to solving the subproblem 
	\begin{align}
	\argmin_v \norm{A - uv^T} \label{eq:factor-sub2}
	\end{align}
	for a given matrix $A$ and now a {\it fixed} vector $u$ (the other subproblem is analogous by transposing $A$ and flipping the roles of $u$ and $v$). Again, when $\norm{\cdot}$ is the 2-norm, and the optimizer of~\eqref{eq:factor-sub2} is just $v^* = Au$, so that the full algorithm is just standard power iteration algorithm for SVD. 
	\begin{algorithm}[!htb]
		\caption{Approximating $\argmin_{u, v} \norm{A- uv^T}$}\label{alg:factor-sub}
		\textbf{Input:} matrix $A \in \R^{m \times n}$
		
		\textbf{Output:} vectors $u, v$.
		
		\begin{algorithmic}[1]
			\STATE make an initial guess for $u$
			\LOOP
			\STATE $v \gets \argmin_v \norm{A - uv^T}$
			\STATE $u \gets \argmin_u \norm{A - uv^T}$
			\ENDLOOP
		\end{algorithmic}
	\end{algorithm}
	
	When $\norm{\cdot}$ is instead the 0-norm, as seen in Algorithm~\ref{alg:factor-sub2}, the optimizer of~\eqref{eq:factor-sub2} is the vector $v$ whose $j$th element is the mode of $A_{ij}/u_i$ over all $i$ for which $u_i \ne 0$. Since the objective function~\eqref{eq:factor-sub} cannot increase during the alternating minimization, and the objective values are integral, Algorithm~\ref{alg:factor-sub} terminates in finitely many iterations at a local optimum. Algorithm~\ref{alg:factor} is an anytime algorithm. In practice, we terminate it when the number of unsuccessful iterations (the number of iterations in which the condition on line 4 is false) exceeds the number of successful iterations. 
	\begin{algorithm}[!htb]
		\caption{Sparse matrix factorization subproblem}\label{alg:factor-sub2}
		\textbf{Input:} matrix $A \in \R^{m \times n}$, vector $u \in \R^m$
		
		\textbf{Output:} vector $v$ minimizing $\norm{A - uv^T}_0$
		
		\begin{algorithmic}[1]
			\STATE $q \gets$ map from indices to lists of real numbers
			\FOR{each $i$ for which $u_i \ne 0$}
			\FOR{each nonzero entry $A_{ij}$ in row $i$ of $A$}
			\STATE append $A_{ij} / u_i$ to $q[j]$
			\ENDFOR
			\ENDFOR
			\STATE $v \gets 0$
			\FOR{each $j$ for which $q[j]$ is nonempty}
			\STATE $M\gets$ mode($q[j]$)
			\STATE count $\gets$ number of times $M$ appears in $q[j]$
			\STATE {\bf if} count $>$ $\norm{u}_0 - \op{len}(q[j])$ {\bf then} $v_j \gets M$
			\ENDFOR
			\label{alg3}
		\end{algorithmic}
	\end{algorithm}
	
	Algorithms~\ref{alg:factor}-\ref{alg:factor-sub2} constitute an approximate algorithm for sparse matrix factorization. It is not exact for two reasons. First, Algorithm~\ref{alg:factor} is greedy: at each step of the loop, it chooses the immediate best rank-1 matrix and greedily appends it to $U$ and $V$. This is not always optimal, and in fact in the worst case can already doom the algorithm to have a trivial approximation factor $\Theta(n)$. Second, Algorithm~\ref{alg:factor-sub} is not exact when $\norm{\cdot}$ is the 0-norm or 1-norm, as expected from our Theorem~\ref{thm:0-norm-hard}. Nevertheless, the  method works remarkably well as we will show experimentally.\footnote{We also experimented with using the 1-norm as a convex relaxation of the 0-norm. Here, the exact solution to~\eqref{eq:factor-sub2} is given by~\citet{mx12}. This seemed to make no difference in practice, so in the experiments we use the 0-norm.} 
	
	\subsection{Initialization}
	For the SVD, the initial guess for $u$ in Algorithm~\ref{alg:factor-sub} is usually chosen to point in a random direction (i.e., $u_i \sim N(0, 1)$ are drawn i.i.d). In our case, this does not work: if we draw $u$ that way, then as long as each column of $A$ has at least two nonzero entries, the mode computation in Algorithm~\ref{alg:factor-sub2} will return $v = 0$ with probability $1$, since $A_{ij}/ u_i$ will be different for each $i$ with probability $1$. This causes the subroutine of Algorithm~\ref{alg:factor} to degenerate, leading to a trivial output. This is troubling for us, since in basically all extensive-form games, $A$ is much sparser than this. Fortunately, one small change yields an initialization that works well. Instead of initializing $u$ to a random unit vector, we initialize $u$ to a random {\it basis} vector $e_i$. This circumvents the above problem, and leads to remarkably strong performance in practice.
	
	\subsection{Implementation with Implicit Matrices}\label{s:implicit}
	A major problem with the above algorithm is that a straightforward implementation of it would store and modify the matrix $A$ in order to factor it. In the setting we are considering, $A$ is often too big to store in memory: the number of nonzero entries in $A$ may be several orders of magnitude greater than the number of rows or columns. In these settings, we would like to be able to implement the algorithm with only {\it implicit} access to $A$. Formally, we  assume access to $A$ via only an immutable oracle that, given an index $i$, retrieves the list of nonzero entries, and their indices, in the $i$th row or $i$th column of $A$.
	
	The immutability of $A$ is the biggest roadblock here. Several changes need to be made to Algorithms~\ref{alg:factor}-\ref{alg:factor-sub2} to accomodate this. First, Line~7 of Algorithm~\ref{alg:factor} is no longer possible.. Thus, Line~3 of Algorithm~\ref{alg:factor}, must be revised to read $\argmin_{u, v} \norm{A - UV^T - uv^T}$, and the matrices $U$ and $V$ must be passed through to Algorithms~\ref{alg:factor-sub} and~\ref{alg:factor-sub2}. On Line~3 of Algorithm~\ref{alg:factor-sub2}, querying the $i$th row of $A - UV^T$ requires a matrix multiplication $U_iV^T$, where $U_i$ is the $i$th row of $U$. 
	
	\subsection{Run-time Analysis}
	The run time of the algorithm depends heavily on the structure of $A$. 
	The worst case run time is $O(\norm{A}_0 n^2)$, since every inner iteration runs in at most quadratic time and removes at least one nonzero entry from $A$.
	In practice it runs dramatically faster than that, and we will now present a rough analysis, valid in most typical cases. For simplicity, assume $A \in \R^{n \times n}$ is square. This doesn't change the analysis in any meaningful way, and makes for easier exposition since we do not need to distinguish when $A$ has been transposed in Algorithm~\ref{alg:factor-sub2}.
	
	As stated above, the run time of the algorithm is dominated by the matrix multiplication $U_i V^T$, which must be performed for every $i$ where $u_i \ne 0$ on the current iteration. On the $r$th outer iteration of the algorithm, $U$ and $V$ will have $r$ columns each; therefore, $V^T \in \R^{r \times n}$, so the matrix multiplication takes time $O(rn)$. We need to perform $\norm{u}_0$ of these per inner iteration. Thus, if the algorithm runs for a total of $R$ outer loop iterations each of which takes $t$ inner-loop iterations, it will take time 
	$$O\qty(t \sum_{r = 1}^R \qty(\norm{u_r}_0 + \norm{v_r}_0) rn).$$
	In practice, the number of inner iterations $t$ per outer iteration is usually very small, say, 3. As an example, if the algorithm correctly factors a matrix of the structure in Example~\ref{ex:rank1}, the $r$th outer loop iteration will find a block of size $O(1/r) \times O(1/r)$. Thus each inner loop iteration just takes time $O(n)$, and there will be $O(n)$ iterations, so that the whole algorithm runs in time $O(n^2) = O(\norm{A}_0)$. 
	
	In most extensive-form games, the payoff matrix $A$ is block diagonal. In this case, running the factorization algorithm on $A$ is equivalent to running it on each of the blocks individually, and has the same run time as running the algorithm on each block separately. Indeed, if $u$ is initialized to a random basis vector $e_i$, the algorithm's entire run, and all its operations---including the critical matrix multiplication $U_iV^T$---will not escape the block to which row $i$ of matrix $A$ belongs. Thus, for example, running the algorithm on a matrix with blocks of size $k \times k$, each of which has the structure of Example~\ref{ex:rank1}, will still take time $O(\norm{A}_0)$. 
	
	\section{Experiments}
	
	\begin{table*}[tb]
		\newcommand{\ra}{\phantom{.00}}
		\caption{Experiments on explicitly specified games. {\it Gap} is the target Nash gap to which {\em LPsparse'} and CFR were run. {\it fnnz} is the total number of nonzero elements that resulted from running our matrix factorization algorithm, reported only when the factorization algorithm had nontrivial effect. {\it Simplex, Barrier, LPsparse'}, and {\it CFR} are the wall-clock times, in seconds, that those four algorithms took to achieve the desired Nash gap. All times greater than 2 hours (7200 seconds) are estimated via linear regression on the log-log convergence plot. Gurobi was time-limited to half an hour (1800 seconds) because each game had at least one method that solved the game well within this limit. Since it is difficult to estimate the convergence rate of Gurobi's solver, Gurobi timeouts are simply indicated with a (T).}
		\vskip 0.15in
		\scriptsize\centering
		\begin{tabular}{lrrrrrrrrrrrr}  
			\toprule
			Game & Gap & $\abs{S_1}+\abs{S_2}$ & $\norm{A}_0$ & fnnz & Simplex & Barrier & LPsparse' & CFR \\
			\midrule
			9-card Leduc poker & .0001\phantom{0} & 5,798 & 30,924 & 13,712 & .5 & {\bf .08} & 7 & 901 \\
			13-card Leduc poker& .0001\phantom{0}  & 12,014 & 95,056 & 31,522 & 2.4 & {\bf .24} & 14 & 1,823 \\
			5x2 battleship m=2 n=1 & .0001\phantom{0} & 230,778 & 33,124 & --- & 8.7 & {\bf .44} & 5 & 2,451 \\
			4x3 battleship m=2 n=1 & .0001\phantom{0} & 639,984 & 82,076 & ---  & 81.0 & {\bf 1.47} & 14 & 4,059  \\
			3x2 battleship m=4 n=1 & .0001\phantom{0} & 3,236,158 & 1,201,284 & --- & (T) & {\bf 16.90} & 659 & 86,284 \\
			3x2 battleship m=3 n=2 & .0001\phantom{0} & 1,658,904 & 3,345,408 & --- & (T) & {\bf 20.22} & 202 & 55,040 \\
			sheriff N=10000 B=100 & .0001\phantom{0} & 1,020,306 & 2,020,101 & --- & {\bf 3.0} & 52.56 & 12 & 7,912  \\
			sheriff N=1000 B=1000 & .0001\phantom{0} & 1,005,006 & 2,003,501 & --- & {\bf 2.8} & 208.35 & 9 & 1,728\\
			sheriff N=100 B=10000 & .0001\phantom{0} & 1,030,206 & 2,020,151 & --- & {\bf 5.2} & 66.71 & 19 & 287  \\
			4-rank goofspiel & .0001\phantom{0} & 42,478 & 11,136 & --- & .7 & {\bf .39} & 42 & 51,857  \\
			5-rank goofspiel & 1.74\phantom{000} & 5,332,052 & 1,574,400 & --- & (T) & {\bf 267.46} & 7,200 & 1,081 \\
			NLH river endgame A & .00684 & 129,222 & 53,585,621 & 481,967 & {\bf 294.9} & (T) & 7,200 & 11,893 \\
			NLH river endgame B & .00178 & 61,062 & 25,240,149 & 229,454 & {\bf 54.4} & (T) & 7,200 & 3,350 \\
			\bottomrule
		\end{tabular}
		\label{tab:experiments}
	\end{table*}
	
	\begin{table*}[tb]
		\caption{Experiments on poker endgames. {\it pot} is the current pot size in big blinds. $\abs{S_1} + \abs{S_2}$ is the total number of sequences across both players. nnz is the number of nonzero entries of the payoff matrix before (first row) and after (second row) the our factorization algorithm is run. The timeout was set to 3600 seconds (1 hour).}
		\vskip 0.15in
		\scriptsize\centering
		\begin{tabular}{lrr|r|rrrrr}  
			\toprule
			&&&&\multicolumn{2}{c}{\quad \quad Factored}& Poker-Specific& Unfactored
			\\
			Endgame & Starting pot (bb) & $\abs{S_1} + \abs{S_2}$  & & Simplex & Barrier & DCFR & Simplex \\
			\midrule
			River 1 & $5.0$ & $95{,}220$
			& time (s) & ${\bf 364}$ & $2{,}116$  & $509$ & $2904$
			\\
			\multicolumn{3}{l|}{factor nnz: $58{,}707{,}847 \to 740{,}218$}
			& memory (MB) & ${\bf 259}$ & $1{,}645$ & $572$ & $5569$
			\\
			\multicolumn{3}{l|}{factor time $145$s, memory $52$MB} & Nash gap (bb)& ${\bf 6.8\times 10^{-8}}$ & $2.8\times 10^{-5}$ & $2.1\times 10^{-4}$ & $6.9\times 10^{-8}$ \\
			
			\midrule
			River 2 & $21.0$ & $68{,}102$
			& time (s) & ${\bf 113}$ & $951$  & $238$ & $830$
			\\
			\multicolumn{3}{l|}{factor nnz: $40{,}817{,}801 \to 662{,}219$}
			& memory (MB) & ${\bf 208}$ & $1{,}126$ & $450$ & $3700$
			\\
			\multicolumn{3}{l|}{factor time $125$s, memory $43$MB} & Nash gap (bb)& ${\bf 8.1\times 10^{-8}}$ & $8.5\times 10^{-7}$ & $2.4\times 10^{-4}$ & $1.0\times 10^{-7}$ \\
			
			\midrule
			River 3 & $5.0$ & $96{,}232$
			& time (s) & ${\bf 410}$ & $1{,}584$  & $591$ & timeout
			\\
			\multicolumn{3}{l|}{factor nnz: $60{,}831{,}748 \to 888{,}608$}
			& memory (MB) & ${\bf 272}$ & $1{,}730$ & $572$ & na
			\\
			\multicolumn{3}{l|}{factor time $181$s, memory $58$MB} & Nash gap (bb)& ${\bf 5.8\times 10^{-8}}$ & $6.3\times 10^{-7}$ & $2.6\times 10^{-4}$ & na\\
			
			\midrule
			River 4 & $10.0$ & $82{,}440$
			& time (s) & ${\bf 231}$ & $1{,}242$  & $389$ & $1936$
			\\
			\multicolumn{3}{l|}{factor nnz: $51{,}332{,}645 \to 781{,}400$}
			& memory (MB) & ${\bf 249}$ & $1{,}433$ & $511$ & $4740$
			\\
			\multicolumn{3}{l|}{factor time $158$s, memory $52$MB} & Nash gap (bb)& ${\bf 1.0\times 10^{-7}}$ & $1.7\times 10^{-6}$ & $2.7\times 10^{-4}$ & $1.2\times 10^{-7}$ \\
			
			\midrule
			River 5 & $5.0$ & $96{,}922$
			& time (s) & ${\bf 210}$ & $1{,}631$  & $366$ & $2120$
			\\
			\multicolumn{3}{l|}{factor nnz: $61{,}078{,}916 \to 816{,}401$}
			& memory (MB) & ${\bf 269}$ & $1{,}735$ & $572$ & $5748$
			\\
			\multicolumn{3}{l|}{factor time $156$s, memory $55$MB} & Nash gap (bb)& ${6.6\times 10^{-8}}$ & $1.7\times 10^{-5}$ & $3.3\times 10^{-4}$ & $\bf 5.9\times 10^{-8}$ \\
			
			\midrule
			River 6 & $36.0$ & $51{,}632$
			& time (s) & ${\bf 38}$ & $516$  & $109$ & $142$
			\\
			\multicolumn{3}{l|}{factor nnz: $27{,}859{,}761 \to 454{,}203$}
			& memory (MB) & ${\bf 164}$ & $848$ & $390$ & $2292$
			\\
			\multicolumn{3}{l|}{factor time $71$s, memory $32$MB} & Nash gap (bb)& ${1.1\times 10^{-7}}$ & $1.4\times 10^{-5}$ & $7.9\times 10^{-4}$ & $\bf 4.8\times 10^{-8}$ \\
			
			\midrule
			River 7 & $37.5$ & $47{,}152$
			& time (s) & ${\bf 21}$ & $708$  & $89$ & $81$
			\\
			\multicolumn{3}{l|}{factor nnz: $23{,}087{,}696 \to 445{,}810$}
			& memory (MB) & ${\bf 159}$ & $770$ & $389$ & $2086$
			\\
			\multicolumn{3}{l|}{factor time $68$s, memory $31$MB} & Nash gap (bb)& ${1.8\times 10^{-7}}$ & $8.0\times 10^{-6}$ & $4.9\times 10^{-4}$ & $\bf 1.7\times 10^{-7}$ \\
			
			\midrule
			River 8 & $25.0$ & $53{,}536$
			& time (s) & ${\bf 51}$ & $644$  & $135$ & $167$
			\\
			\multicolumn{3}{l|}{factor nnz: $30{,}197{,}553 \to 488{,}937$}
			& memory (MB) & ${\bf 167}$ & $792$ & $389$ & $2702$
			\\
			\multicolumn{3}{l|}{factor time $84$s, memory $34$MB} & Nash gap (bb)& ${1.0\times 10^{-7}}$ & $\bf 9.3\times 10^{-9}$ & $2.6\times 10^{-4}$ & $6.8\times 10^{-8}$ \\
			
			\midrule
			Small Turn & $200.0$ & $352{,}800$
			& time (s)  & $3{,}241$ & ${\bf 482}$  & $726$ & timeout
			\\
			\multicolumn{3}{l|}{factor nnz: $96{,}450{,}855 \to 2{,}680{,}527$}
			& memory (MB) & $887$ & $1{,}545$ & ${\bf 133}$ & na
			\\
			\multicolumn{3}{l|}{factor time $244$s, memory $193$MB} & Nash gap (bb)& ${\bf 2.3\times 10^{-8}}$ & $3.3\times 10^{-6}$ & $1.0\times 10^{-6}$ & na \\
			
		\end{tabular}
		\label{tab:poker-experiments}
	\end{table*}
	
	We compared state-of-the-art commercial implementations~\cite{gurobi} of the common LP solving algorithms (simplex, dual simplex, and barrier) and our modified version of the $O(\log^2(1/\eps))$ {\it LPsparse} algorithm~\citep{yen15} (which we call {\it LPsparse'}), combined with our factorization algorithm, to the newest, fastest variants of CFR~\cite{bs19}. 
	
	\subsection{Experiments with All Solvers}
	
	In the first set of experiments, we studied the setting where the payoff matrix $A$ is given explicitly. In this setting, the factorization algorithm can be allowed to modify $A$, and CFR variants must load the whole matrix $A$ into memory. In this experiment, we use the   game-independent CFR implementation 
	built in the Rust programming language for speed.  
	In each game, the largest entry of the payoff matrix in absolute value, that is, $\norm{A}_\infty$, was normalized to be $1$. 
	We ran {\it LPsparse'} four times on each game; in particular, for each combination of (i) which player is chosen to be player $x$ in~\eqref{eq:lp}, in other words, whether~\eqref{eq:lp} is solved via the primal or dual;  and (ii) choice of inner iteration algorithm (RC or PG). 
	We tested four different variants of CFR: DCFR[$\infty, -\infty, 1$] (``CFR+''), DCFR[$\infty, -\infty, 2$] (``CFR+ with quadratic averaging''), DCFR[$1.5, 0, 2$] (``DCFR''),  DCFR[$1, 1, 1$] (``LCFR''). These variants are introduced and analyzed in depth by~\citet{bs19} and represent the current state of the art in large-scale game solving. The best of those variants for each game is shown in Table~\ref{tab:experiments}.
	We ran {\it LPsparse'} and CFR to target precision (Nash gap) $10^{-4}$, or for $2$ hours, whichever threshold was hit first. We ran primal and dual simplex to optimality (machine precision), and barrier with default settings except crossover off. We ran all solvers on a single core. The games that we tested on are standard benchmarks; they are described in the appendix.
	
	On most games, all LP solvers outperformed CFR. This marks, to our knowledge, the first time that LP (or, indeed, any fundamentally different algorithm) has been shown to be competitive against leading CFR variants on large games. 
	
	The matrix factorization algorithm performs remarkably well in practice when it needed to. On 9-card and 13-card Leduc poker, it led to a compression ratio of 2-3. More impressively, the algorithm compresses both no-limit endgames by a factor of more than 100. This brings savings of nearly the same factor in convergence rate in both games, and enables the LP algorithms to be competitive against the CFR variants in these large games. On payoff matrices that are already sparse, the factorization algorithm fails to find a sparse factorization, and terminates immediately.
	
	On a few games, the choice of which player to make the $x$ player in LP~\eqref{eq:lp};\, that is, the choice between primal and dual solves, made a significant difference. For example, in the {\it sheriff} family of games, setting $x$ to be the smuggler yields much better results. This is because the optimal strategy in the sheriff games is very sparse for the smuggler. Indeed,~\citet{yen15} make note of the fact that their algorithm performs significantly better when the optimal primal solution is sparse, since in this case the inner loop does not need to loop over the entire constraint matrix $A$. 
	
	\subsection{Experiments on No-limit Texas Hold'em Endgames}
	
	In the experiment described above, Gurobi's LP solvers consistently outperformed {\it LPsparse'} despite the theoretical guarantees of the latter. Thus, in the second set of experiments, we focus on Gurobi and DCFR.
	
	The implicit implementation of our factorization algorithm (Section~\ref{s:implicit}) allows us to scale our method to larger games than previously possible. We hence ran experiments testing this implementation on heads-up no-limit Texas Hold'em poker endgames encountered by the superhuman agent {\it Libratus}~\cite{bs18}. To align with~\citet{bs19}, we used a simple action abstraction: the bets are half-pot, full-pot, and all-in, and the raises are full-pot and all-in. All results are expressed in terms of the standard metric, namely big blinds (bb). The starting stacks are 200 big blinds per player as in the Libratus match against humans. We tested on eight real river endgames (i.e., endgames that start on the fourth betting round) and a single manually-generated small turn endgame (i.e., endgames that start on the third betting round) where the pot already has half of the players' wealth, so only a single additional bet or raise is possible.
	
	In this experiment we used an optimized poker-specific C++ implementation of DCFR.
	This implementation includes optimizations such as those of~\citet{jwbz11}, which shave an $O(k)$ factor off the runtime of CFR, where in the case of Texas hold'em poker, $k = 1326$ is the number of possible hands a player may have, and~\citet{bs15}, which prune game lines that are dynamically determined not to be part of the optimal solution. For the LP solver, we use Gurobi's simplex and barrier methods. Both primal and dual simplex were run, and only the better of the two results is shown in Table~\ref{tab:poker-experiments}. We also tested Gurobi without the factorization algorithm. In this case, we do not include results for the barrier method, because it timed out or ran out of memory in all the cases. All algorithms were again restricted to a single core.
	DCFR was run for the amount of time taken by the fastest LP variant that used factoring. For example, if Gurobi took $200$ seconds to solve a game, and the factorization algorithm took $100$ seconds, CFR was run for $300$ seconds. The results are in Table~\ref{tab:poker-experiments} and representative convergence plots showing anytime performance are in the appendix. 
	
	The factorization algorithm reduced the size of the game by a factor of 52--80 and the resulting payoff matrix had density (i.e., nonzeros divided by rows plus columns) 7.8--9.5. This is expected: poker payoff matrices are block diagonal, where the blocks are $k \times k$ and rank one with the lower-triangular half negated. Thus, they basically have the structure of Example~\ref{ex:rank1}, in which we saw a compression from density $k \approx 2^{10}$ to density $\log k \approx 10$, which is exactly the compression we are seeing here. 
	
	\begin{figure*}[!htb]
		\centering
		\includegraphics[width=0.9\columnwidth]{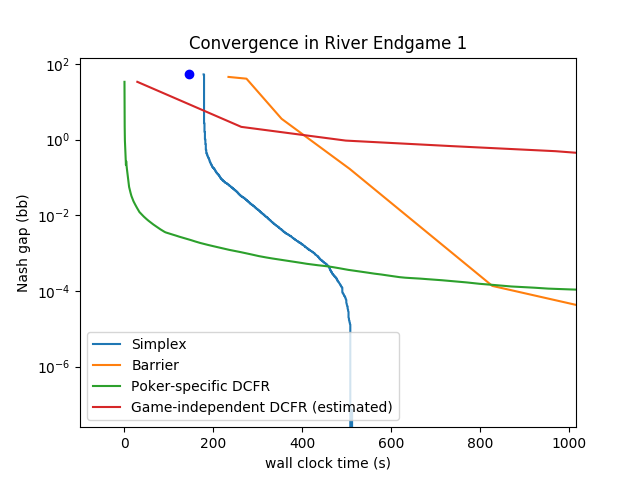}
		\includegraphics[width=0.9\columnwidth]{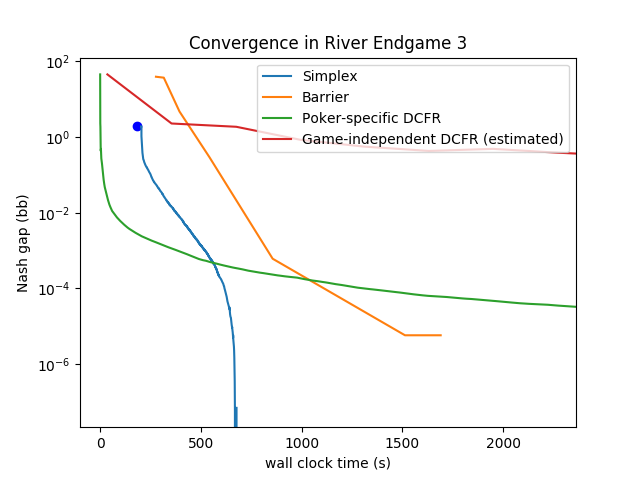}
		\includegraphics[width=0.9\columnwidth]{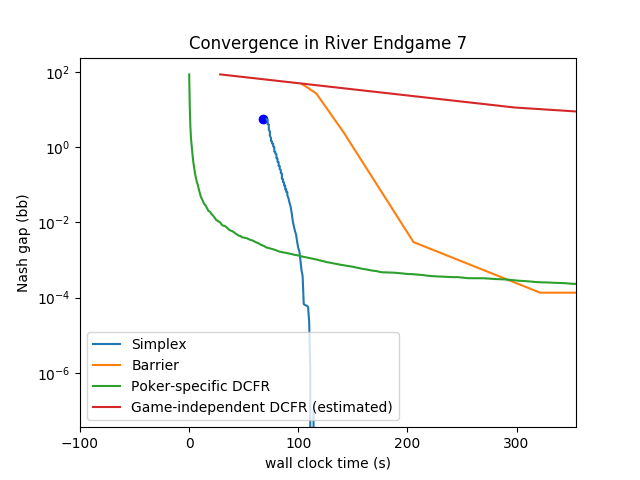}
		\includegraphics[width=0.9\columnwidth]{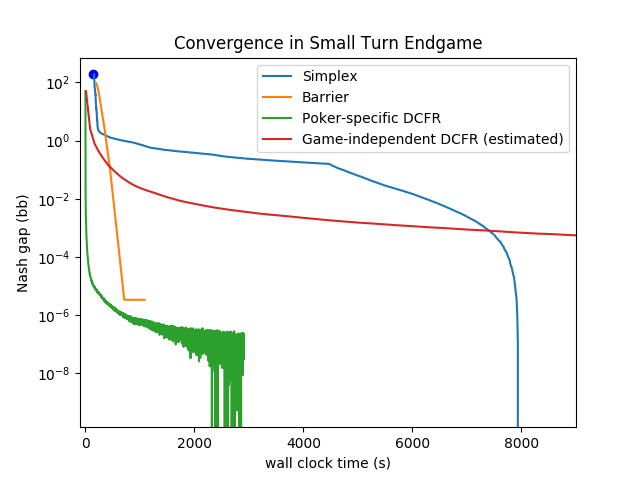}
		\vspace{-0.9\baselineskip}
		\caption{Convergence plots on representative endgames. DCFR is plotted against the best-performing LP algorithm. The blue dot represents the time taken by the factorization algorithm, and the space between the blue dot and the start of the blue line is the time taken for the algorithm to initialize the algorithm, and then find a feasible solution (simplex) or run one iteration (barrier). The drop below zero of the simplex plot is due to a quirk of Gurobi's objective value reporting, and can most likely be safely ignored. The drop below zero of the poker-specific DCFR in the small turn endgame is due to machine precision issues, and once again can be ignored.}
		\label{fig:plots}
	\end{figure*}
	
	The DCFR implementation we tested against is especially optimized to solve no-limit turn endgames. It thus may have some inefficiencies when handling river endgames. We estimate that these inefficiencies lose a factor of approximately 20 in time and space on river endgames relative to a river-optimized implementation. However, importantly, these inefficiencies pale in comparison to the speedups gained by game-specific poker speedups (e.g.,~\citealp{jwbz11}), which save a factor of approximately $k = 1326$ in time (but not space). This strongly suggests that our method would be significantly faster than any non-game-specific implementation of CFR or any modern variant. Furthermore, the memory usage of simplex, after factorization, is only a factor of $\log k \approx 10$ worse than the game-specific CFR (which stores the constraint matrix implicitly) in the case of all these endgames, which means it is often practical to use LP solvers even on extremely large games with dense payoff matrices, as long as the constraint matrix is factorable.
	
	Since primal simplex and dual simplex give respectively only primal-feasible and dual-feasible solutions, anytime performance of simplex is measured by running both simultaneously, and measuring the Nash gap between the current primal and dual solutions at each time checkpoint, using Gurobi's reported objective values. While Gurobi does not allow retrieval of these anytime solutions when its presolver is turned on, in principle they can be retrieved easily using the presolver's mapping, which unfortunately Gurobi does not expose to the end user. The convergence plots in Figure~\ref{fig:plots} show roughly what we would expect. CFR has a very stable convergence curve (until it hits too high precision, at which point numerical stability issues start kicking in, and the convergence plot looks weird). The LP solvers start out slow (especially due to the sometimes nontrivial time requirements of the factorization algorithm) but catch up with and often eventually exceed the performance of DCFR, before again very often stopping due to numerical issues. Even on turn endgames, LP algorithms consistently outperform a hypothetical non-game-specific implementation of DCFR---which we define to be 500 times slower than the poker-specific DCFR---due to the additional factor of $k \approx 1326$ in the density of the payoff matrix, and hence the additional cost of the gradient computation in DCFR. 
	\section{Conclusion and Future Research}
	
	We presented a matrix factorization algorithm that yields significant reduction in sparsity. We showed how the factored matrix can be used in an LP to solve zero-sum games. This reduces both the time and space needed by LP solvers. On explicitly represented games, this significantly outperforms the prior state-of-the-art algorithm, DCFR. It also made LP solvers competitive on large games that are implicitly defined by a compact set of rules---even against an optimized game-specific DCFR implementation. 
	There are many interesting directions for future research, such as (1) further improving the factorization algorithm, (2) investigating the explicit form of an optimal factorization in special cases, and (3) parallelizing the factorization algorithm.
	\section*{Acknowledgements}
	This material is based on work supported by the National Science Foundation under grants IIS-1718457, IIS-1617590, IIS-1901403, and CCF-1733556, and the ARO under awards W911NF1710082 and W911NF2010081.
	\bibliographystyle{icml2020}
	\bibliography{main}
	
	\clearpage
	\onecolumn
	\appendix
	\section{Proof of Theorem \ref{thm:normalization}}
	The key to the proof is to bound how much this naive normalization changes the point $x$. Let $(x^*, z^*)$ be the result of projecting $(x, z)$ into the optimal set $S$.
	\begin{lemma*}
		Let $x'$ be the result of normalizing $x$ according to the given scheme, and $i$ be an information set at depth $d$ (with the root defined to be at depth $0$. Then we have $\abs{x'_i - x^*_i} \le \eps d\sqrt{n}.$
	\end{lemma*}
	\begin{proof}
		By induction on the sequence-form strategy tree for player $x$, starting at the root. At the root node $i=0$, the claim is clearly true because $x_0 = 1$ in any feasible solution $x$. Now consider any information set with parent $x_{i_0}$ and children $x_i := (x_{i_1}, \dots, x_{i_k})$ at depth $d$. From the theorem statement, we have $\norm{x_i - x^*_i}_2 \le \eps$, and since $x^*$ is feasible, we have $\sum_{j=1}^k x_{i_k}^* = x_{i_0}^*$. It follows that $$\abs{\sum_{j=1}^k x_{i_k} - x_{i_0}'}\le \sum_{j=1}^k \abs{ x_{i_k} - x_{i_k}^*} + \abs{x'_{i_0} - x^*_{i_0}} \le \eps \sqrt{k} + \eps(d-1) \sqrt{n} \le \eps d \sqrt{n}$$ by triangle inequality and inductive hypothesis, and noting that $k \le n$. But the normalization acts by picking $x_i'$ so that $\sum_{j=1}^k x_{i_k}' = x_{i_0}'$, and it moves all the $x_{i_k}$s in the same direction; thus, each one can move by at most $\eps d \sqrt{n}$, completing the induction.
	\end{proof}
	With this lemma in hand, we now prove the theorem.
	\begin{proof}[Proof of Theorem]
		Since $d \le n$ (each depth must have at least one information set), it follows from the lemma that $\norm{x' - x^*}_2 \le \eps n^2$. But the best response function $\min_y x^T A y$ (with feasibility constraints on $y$) is a pointwise minimum of Lipschitz functions $x^T v$ for each $v = Ay$ and $y$ feasible, hence itself Lipschitz, with Lipschitz constant
		$$\max_y \norm{Ay}_2 \le \max_y \norm{Ay}_1 \le \norm{A}_1 \max_y \norm{y}_\infty = \norm{A}_1 \le n^2 \norm{A}_\infty.$$
		where $\norm{A}_1$ is the sum of the magnitudes of the nonzero entries of $A$. The desired theorem follows.
	\end{proof}
	\par~
	\section{Another Example of the Utility of Sparse Factorization}
	\begin{exampl}\label{ex:easy}
		Let $A$ be the $n \times (n+1)$ matrix given by $A = \mqty[I_n & 0] + \mqty[0 & I_n]$, where $I_n$ is the $n \times n$ identity, and $0$ is a column vector of zeros. So, $A$ is the matrix whose $(i, j)$ entry is $1$ exactly when $j = i$ or $j = i+1$. By direct computation, the SVD of this matrix is $A = U\Sigma V^T$ where $U$ and $V$ are {\it fully dense}, and the SVD is unique (in the usual sense, that is, up to signs and permutations) since all the singular values are. Thus, taking an SVD would have the result of {\it increasing} the number of nonzeros from $2n$ to $\Theta(n^2)$, which is the opposite of what we want. Thus, although in this case there will not be a good sparse factorization, using  SVD make the problem worse.
	\end{exampl}
	\par~
	\section{Benchmark Games in Experiment 1}
	
	We tested on the following benchmark games from the literature:
	\begin{itemize}
		\item {\it Leduc poker} \citep{Southey05:Bayes} is a small variant of poker, played with one hole card and three community cards.
		\item {\it Battleship} \citep{Farina19:Correlation} is the classic targeting game, with two parameters: $m$ is the number of moves (shots) a player may take, and $n$ is the number of ships on the board. All ships have length 2. A player scores a point only for sinking a full ship.
		\item {\it Sheriff} \citep{Farina19:Correlation} is a simplified Sheriff of Nottingham game, modified to be zero-sum, played between a {\it smuggler} and a {\it sheriff}. The smuggler selects a {\it bribe amount} $b \in [0, B]$ and a number of illegal items $n \in [0, N]$ to try to smuggle past the sheriff. The sheriff then decides whether to inspect. If the sheriff does not inspect the cargo, then the smuggler scores $n - b$. If the sheriff inspects and finds no illegal items ($n = 0$), then the smuggler scores $3$. If the sheriff inspects, and $n > 0$, then the smuggler scores $-2n$. The smuggler has far more sequences than the sheriff in this game.
		\item {\it No-limit hold-em (NLH) river endgames} are endgames encountered by the poker-playing agent Libratus~\cite{bs18}, using the action abstraction used by Libratus. They both begin on the last betting round, when all five community cards are known. The normalization of $\norm{A}_\infty = 1$ means that in these endgames, a Nash gap of 1 corresponds to 0.075 big blinds. Due to the explicit storage of the payoff matrix in this experiment, only extremely small no-limit endgames can be tested. In particular, {\it endgame A} here is the same as {\it endgame 7} in the next experiment (with a finer abstraction), and {\it endgame B} is the same as {\it endgame A} except with the starting pot size doubled to make the game smaller.
	\end{itemize}
	
\end{document}